\documentclass{scrartcl}

\usepackage{geometry}

\usepackage{blindtext}
\usepackage[T1]{fontenc}
\usepackage[utf8]{inputenc}

\usepackage{amsmath, amsfonts, amsthm, amssymb}
\usepackage{braket, nicefrac}
\usepackage{tikz}
\usetikzlibrary{automata,arrows,positioning,calc}
\usepackage{dsfont}
\usepackage{relsize}
\usepackage{tikz-cd}

\usepackage{siunitx}

\usepackage{enumitem, multicol}
\usepackage{multirow}

\usepackage{svg}
\DeclareUnicodeCharacter{3B1}{\ensuremath{\alpha}}
\DeclareUnicodeCharacter{3B2}{\ensuremath{\beta}}
\usepackage{graphicx, float}  
\usepackage{keystroke}
\usepackage{pgfplots}\usepgfplotslibrary{units}\pgfplotsset{compat=1.16}
\newtheorem{theorem}{Theorem}
\newtheorem*{remark}{Property}

\usepackage{pgfplots}         
\pgfplotsset{compat=1.16}
\usetikzlibrary{automata,
                arrows.meta,    
                positioning,    
                quotes}         
\usepgfplotslibrary{fillbetween}
\tikzset{node distance=4.5cm,   
         every state/.style={   
                semithick,
                fill=gray!10},
         initial text={},       
         double distance=4pt,   
         every edge/.style={    
                draw,
                semithick,
                -Stealth,       
                auto},
         bend angle=15          
         }

\usepackage{hyperref}           

\usepackage{setspace}

\usepackage{hyperref}

\graphicspath{{graphics/}{Graphics/}{./}}

\renewenvironment{abstract}{
    \begin{center}
    \textbf{Abstract}
    \vspace{0.5cm}
    \par\itshape
    \begin{minipage}{0.8\linewidth}}{\end{minipage}
    \noindent\ignorespaces
    \end{center}
}

\newenvironment{keywords}{
    \begin{center}
    \textbf{Keywords}
    \vspace{0.5cm}
    \par
    \begin{minipage}{0.8\linewidth}}{\end{minipage}
    \noindent\ignorespaces
    \end{center}
}

\newenvironment{preface}{
    \begin{center}
    \textbf{Preface}
    \vspace{0.5cm}
    \par
    \begin{minipage}{0.8\linewidth}}{\end{minipage}
    \noindent\ignorespaces
    \end{center}
}

\newenvironment{acknowledgements}{
    \begin{center}
    \textbf{Acknowledgements}
    \vspace{0.5cm}
    \par
    \begin{minipage}{0.8\linewidth}}{\end{minipage}
    \noindent\ignorespaces
    \end{center}
}

\newcommand{\lsum}{\mathop{\mathlarger{\mathlarger{\mathlarger{\sum}}}}}

\begin{document}
\begin{titlepage}
	\centering
	\includegraphics[width=0.6\textwidth]{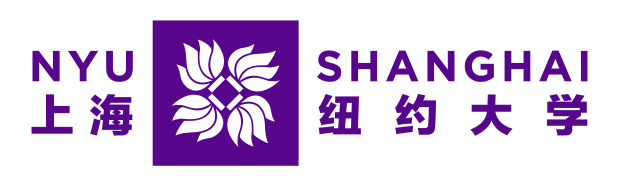}\par
	\vspace{2cm}
	{\scshape\LARGE Computer Science \par}  
	{\scshape\LARGE \& \par}                
	{\scshape\LARGE Data Science \par}      
	\vspace{1cm}
	{\scshape\Large Capstone Report - Spring 2022\par}
	\vfill
	
	{\huge\bfseries  ASIR: Robust Agent-based Representation Of SIR Model \par}
	\vfill
	
	{\Large\itshape Boyan Xu\\}\par
	\vspace{1.5cm}

	\vfill
	supervised by\par
	Olivier Marin

	\vfill
\end{titlepage}

\newpage

\begin{preface}

As a student double majoring in Computer Science and Data Science, I have always been attracted by the idea behind computer simulation: modeling the real world in programming language, making statistical inference about the real world. I said yes immediately when professor Olivier Marin told me if I was interested in on building some bridge between compartmental model (which is based on ordinary differential equation) and agent-based model (which is based on computer simulation.) We only have a vague direction at the beginning, but as we went further in the literature, we were more and more certain about we would like to contribute to. Finally, we proposed the agent-based $ASIR$ that can robustly reproduce the infection curve predicted by the compartmental SIR Model. We are happy to see that $ASIR$ is inspirational for epidemiologists who wish to quickly transform a calibrated SIR model into an agent-based model that retains its prediction without another round of calibration.

\end{preface}

\vspace{1cm}

\begin{acknowledgements}

First, I would like to express my deepest gratitude to my parents. They have always been supportive, open-minded, tolerant on every decision I made for myself. While NYU Shanghai's education makes me an intellectually cultivated person, it is their education that gives me my integrity of personality. Then I would like to acknowledge every professor I've met at NYU Shanghai and NYU. They are all professional instructors and considerate professors. Especially, I would like to sincerely acknowledge Professor Olivier Marin for patiently instructing me how to read, think and write academically. I cannot complete my capstone project without his unconditional support and comprehensive instruction. Thank Professor Guyue Liu for encouraging me, and supporting my capstone project during the depressing lockdown. Thank Professor Bruno Abrahao for supervising my research in data science. Finally, I would like to thank all the peers I met at NYU Shanghai, they constitute my undergraduate life. Thank  my "dude" Quang Luong for being my best CS professor. Thank my girlfriend Helen Zhang for always staying by my side.

\end{acknowledgements}

\newpage

\begin{abstract}
 
Compartmental models (written as $CM$) and agent-based models (written as $AM$) are dominant methods in the field of epidemic simulation. But in the literature there lacks discussion on how to build the \textbf{quantitative relationship} between them. In this paper, we propose an agent-based $SIR$ model: $ASIR$. $ASIR$ can robustly reproduce the infection curve predicted by a given SIR model (the simplest $CM$.) Notably, one can deduce any parameter of $ASIR$  from parameters of $SIR$ without manual tuning. $ASIR$ offers epidemiologists a method to transform a calibrated $SIR$ model into an agent-based model that inherit $SIR$'s performance without another round of calibration. The design $ASIR$ is inspirational for building a general quantitative relationship between $CM$ and $AM$.


\end{abstract}
\vspace{1cm}

\begin{keywords}
\centering
        \textbf{Computer Simulation; Epidemic Simulation}
\end{keywords}

\newpage

\doublespacing
\tableofcontents
\singlespacing

\newpage

\doublespacing

\section{Introduction}


Compartmental models (written as $CM$) and agent-based models (written as $AM$) are dominant methods in the field of epidemic simulation.\cite{Grimm_Berger_Bastiansen_Eliassen_Ginot_Giske_Goss-Custard_Grand_Heinz_Huse_et_al._2006}. $CM$ capture the \textbf{population level} dynamics by a set of ordinary differential equations. $AM$ capture the \textbf{individual level} dynamics by an agent-based programming environment.

 $CM$ and $AM$ have complementary nature. $CM$ are easy to calibrate but have less flexible parameter space to apply a priori; $AM$ have a flexible parameter space to apply a priori, but are hard to calibrate. In the current literature, there lacks discussion on developing the \textbf{quantitative relationship} between $CM$ and $AM$. Driven by this fact, we wish to bridge the gap between $CM$ and $AM$ by finding an $AM$ that \textbf{robustly reproduces the infection curve} predicted by a $CM$.

We start from the $SIR$ model (the simplest $CM$): where $P_{sir}$ is the set of all parameters; and functions $S_{sir}(t)$, $I_{sir}(t)$, $R_{sir}(t)$ are the population size of being susceptible, infected and recovered with respect to time $t$. We propose an agent-based SIR model, $ASIR$, that achieves the following interesting properties:
\begin{enumerate}
\item $P_{asir}$ only depends on $P_{sir}$. Any parameter $p \in P_{asir}$ can be deduced from $P_{sir}$, i.e. can be written as a determinate expression of $\{p_1, p_2, ..., p_k\} \subset P_{sir}$.
\item $ASIR$ robustly reproduces the infection curve predicted by $SIR$. $ASIR$ is expected to predict the same $S(t), I(t), R(t)$  as $SIR$, i.e. $\forall t: \mathbb{E}\big [ S_{asir}(t) \big] = S_{sir}(t)$; \ $\mathbb{E}\big[ I_{asir}(t) \big] = I_{sir}(t)$; \ $\mathbb{E}\big[ R_{asir}(t)\big] = R_{sir}(t) $
\end{enumerate}

We validate $ASIR$'s properties by giving: 1). a proof of robustness, 2). two implementations in \textbf{GAMA} and \textbf{Agents.jl}.

\section{Related Work}

$SIR$ model is the simplest compartmental model ($CM$). In ~\cite{sir-intro}, the authors give us an overview of the design behind the compartmental model. From this work, we learned that the core idea of $CM$ is to use a set of ordinary differential equations to model the population infection and recovery, and use parameters to control their rate.

In ~\cite{abm-intro}, the authors give us an overview of the agent-based simulation's application in the field of epidemiology. From this work, we learned that representation of "space" is what distinguish agent-based model ($AM$) from $CM$, and drew our attention to the design of agents' \textbf{Move} behavior.

In ~\cite{sim-zombie}, the authors describe a general method for the conversion of an equation-based model to an agent-based simulation. Their method was not built on solid mathematics, but the discussion about the relationship between population behavior and individual behavior has greatly inspired our idea behind $ASIR$. Our model can be seen as translating their rough ideology into rigorous proof in mathematics.

~\cite{gama} and  ~\cite{Agents.jl}  are the multiagent programming environments we use to implement $ASIR$. Their design is the direct source of our perception of what is agent-based simulation. $ASIR$ has been influenced by the concepts of "Agent," "Step," and "Map" that were implemented by  ~\cite{gama} and  ~\cite{Agents.jl}.

~\cite{markov} is the book we used as a reference for the necessary condition for the existence of a Markov chain's stationary distribution.

\section{Solution}

In this section, we introduce $ASIR$ in the following order: 1). idea behind, 2). model specification, and 3). proof of robustness.

\subsection{Idea Behind $ASIR$}
Let us first briefly recap the design of the $SIR$ model. 
\break
\begin{figure}[h]
	\begin{center}
		\includegraphics[scale=0.618]{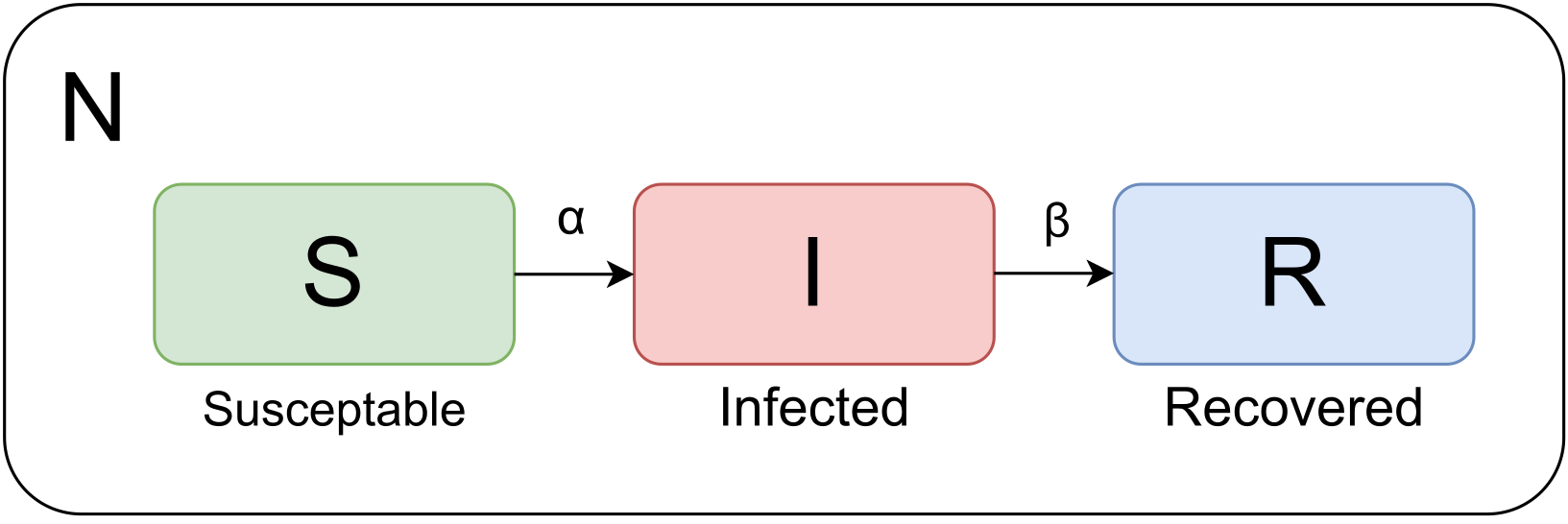}
	\end{center}
	\caption{Diagram of SIR model}
	\label{graph:sir-diagram}
\end{figure}

The $SIR$ model consists of two parameters: $\{\alpha, \beta \}$ and three ordinary differential equations:
\begin{align}
\frac{d S}{d t}&=-\frac{\alpha S I}{N} \label{sir1}\\
\frac{d I}{d t}&=\frac{\alpha S I}{N}-\beta I \label{sir2}\\
\frac{d R}{d t}&=\beta I \label{sir3}
\end{align}
$N$ is the total population size. $S$ is the susceptible population size. $I$ is the infected population size. $R$ is the recovered population size. $\alpha$ controls the transition speed from \textbf{S}usceptible into \textbf{I}nfected. $\beta$ controls the transition speed from \textbf{I}nfected into \textbf{R}ecovered.

Equation \ref{sir1} models the transition speed of \textbf{S}usceptible population size. Equation \ref{sir2} models the transition speed of \textbf{I}nfected population size. Equation \ref{sir3} models the transition speed of \textbf{R}ecovered population size.

The intuitions behind $ASIR$ are:

\begin{enumerate}
\item \textbf{Population infection} is an integral of \textbf{individual infection}.
\item \textbf{Population recovery} is an integral of \textbf{individual recovery}.
\end{enumerate}

To translate these intuitions into mathematics, we adopt the theory of probability by treating population infection/recovery as the joint distribution of individual infection/recovery.

The core ideas behind $ASIR$ are:

\begin{enumerate}
\item Model individual infection/recovery as \textbf{mutually independent and identically distributed random events}.
\item Use parameters to control the \textbf{event probability}.
\item The \textbf{transition speed} on population-level = the \textbf{expected value} of the \textbf{integral of event probability} on individual-level.
\end{enumerate}

Notably, since a \textbf{S}usceptible individual must be infected by an \textbf{I}nfected individual, an individual infection at time $t+1$ is \textbf{conditional} on another individual's infection at time $t$ (or ahead of $t$ ). To guarantee \textbf{independence} between each individual's infection, we model the movement of every agent (or individuals, we are using these words interchangeably) using the same transition matrix $T_{\text{map} }$. 
The matrix below shows a simple example where our $\text{map}$ consists of three locations:

\begin{equation}
\mathbf{T_{\text{map}}} = 
    \bordermatrix{ & \text{Store} & \text{School} & \text{Stop} \cr
      \text{Store} & 0.5 & 0.3 & 0.2 \cr
      \text{School}  & 0.3 & 0.3 & 0.4 \cr
      \text{Stop} & 0.2 & 0.4 & 0.4 } \qquad
\label{eq:tranmat}
\end{equation}
\break

In this example, coordinate $T_{mn}$ is the probability of moving from $ m \text{ to } n$. We focus on the period after every agent's trajectory reaches $T_{\text{map} }$'s stationary distribution. We discuss why stationary distribution is critical for robustness later in this section.

In the following subsection, we introduce $ASIR$'s detailed specification in the following order: 1). agents' state, 2). agents' behavior, 3). model's parameter setting.

\subsection{Model Specification}

Each agent's state can be written as a 3-element tuple: $\Big(\text{Timestamp}, \text{Health}, \text{Position} \Big)$. "An agent $a_1$ has state $\Big(8, I, \text{School} \Big)$" translates as: "at the 8-$th$ timestamp, $a_1$ is being \textbf{I}nfected at $\text{School}$." In this paper, we will use:

\begin{enumerate}
\item $\big(a_k, h_t\big)$ or $H^{t}_{a_k}$ to denote an agent's health at timestamp $t$,
\item $\big(a_k, p_t\big)$ or $P^{t}_{a_k}$ to denote an agent's position at timestamp $t$,
\item $\big(a_k , h_t, p_t \big)$ or $\big( H^{t}_{a_k}, P^{t}_{a_k}\big)$ to denote an agent's state at timestamp $t$.

\end{enumerate}

\begin{table}[htp]
\begin{center}

\begin{tabular}{ |p{5.5cm}||p{2cm}|p{2cm} |  }
 \hline
 Meaning &Form 1 &Form 2\\
 \hline
 &    & \\
 an agent's health at timestamp $t$   & $\big(a_k, h_t\big)$    &$H^{t}_{a_k}$\\
 &    & \\
 \hline
 &    & \\
 an agent's position at timestamp $t$  & $\big(a_k, p_t\big)$   &$P^{t}_{a_k}$\\
 &    & \\
 \hline
 &    & \\
 an agent's state at timestamp $t$  & $\big(a_k , h_t, p_t \big)$   & $\big( H^{t}_{a_k}, P^{t}_{a_k} \big)$\\
 &    & \\
\hline
\end{tabular}
\caption{\label{tab:state-reference} Symbol reference for agent state.}
\end{center}
\end{table}

At timestamp $t$, each agent $a_k$ has three (potential) behaviors: 
\begin{enumerate}
\item \textbf{Move}. $a_k$ moves from  $p_{t-1}$ to $p_t$ (which can be the same position as $p_{t-1}$). Written as:
\begin{equation}
    P_{a_k}^{t-1}  \to  P_{a_k}^{t}
\label{eq:move}
\end{equation}
As mentioned in the introduction, we model the movement of every agent using the same transition matrix $T_{\text{map} }$. The trajectory of every agent forms a Markov chain as the example below shows:
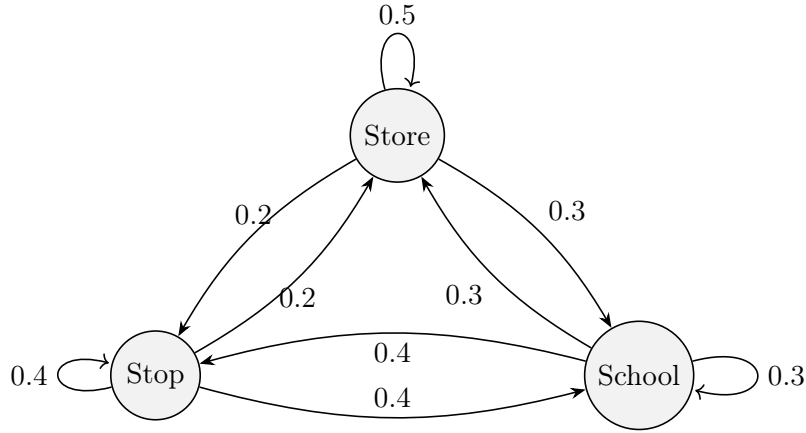
\begin{figure}[htp]
\centering
\begin{tikzpicture}[]
\node[state] (s1) {Store};
\node[state, below right of=s1] (s2) {School};
\node[state, below left of=s1]  (s3) {Stop};

\draw   (s1) edge[loop above] node{0.5}  (s1)
        (s1) edge[bend left]  node{0.3}  (s2)
        (s1) edge[bend right, above] node{0.2}  (s3)
        (s2) edge[bend left]  node{0.3}  (s1)
        (s2) edge[loop right] node{0.3}  (s2)
        (s2) edge[bend right] node{0.4}  (s3)
        (s3) edge[bend right, below] node{0.2}  (s1)
        (s3) edge[bend right] node{0.4}  (s2)
        (s3) edge[loop left]  node{0.4}  (s3);
\end{tikzpicture}
\caption{Markov chain of the sample $T_{\text{map}}$ \ref{eq:tranmat} }
	\label{graph:markov-eample}
\end{figure}

We insist on focusing on the period after every agent's position reaches $T_{\text{map} }$'s stationary distribution, because the stationary distribution offers us the following critical property to deduce $ASIR$'s robustness:

\begin{remark}
When every agent's position reaches the $T_{\text{map} }$'s stationary distribution, the probability that any two agents \textbf{become neighbor} at time $t$  (i.e. stay at the same position $p_t$) equals to a constant $\mathds{P}(\text{meetup})$.  $\mathds{P}(\text{meetup})$ is completely determined by $T_{\text{map} }$.
\end{remark}

The proof is trivial. An intuition is that agents' locations are mutually independent and identically distributed, therefore $\forall j,k,m,n, \ \mathds{P}(a_k \text{ meets } a_j) = \mathds{P}(a_m \text{ meets } a_n)$

\item \textbf{Turn infected}. When $a_k$ was \textbf{S}usceptible before moving to position $X$, it has a chance to turn infected when there is an "\textbf{I}nfected neighbor" at $X$,  or more precisely, $\exists_{\ a_j \neq a_k}  H_{a_j}^t = I, P^{a_j}_t = P^{a_k}_t $. Written as: 
\begin{equation}
   \Bigg ( H_{a_k}^{t-1}\to H_{a_k}^{t}  =   S\to I \Bigg \vert  \exists_{\ a_j \neq a_k}  H_{a_j}^t = I, P^{a_j}_t = P^{a_k}_t \Bigg )
   \label{eq:infect1}
\end{equation}

or simply:

\begin{equation}
   \Bigg ( H_{a_k}^{t-1}\to H_{a_k}^{t}  =  S\to I  \Bigg \vert a_k  \text{ has an infected neighbor at t} \Bigg  )
   \label{eq:infect2}
\end{equation}\begin{figure}[h]
	\begin{center}
		\includegraphics[scale=0.618]{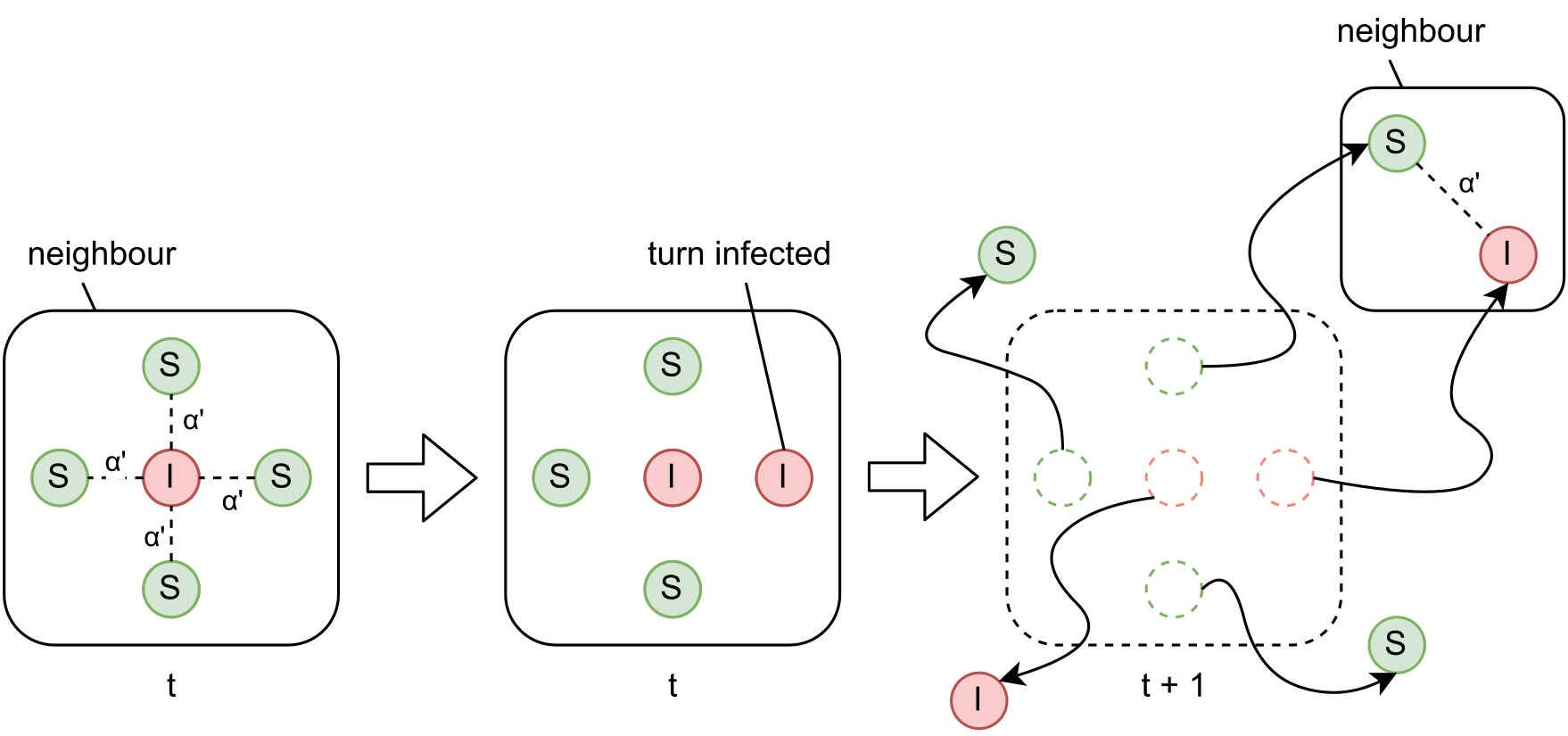}
	\end{center}
	
\caption[Individual infection process in ASIR model; $\alpha$' is the infection probability]
    {\tabular[t]{@{}l@{}}Individual infection process in ASIR model. \\ $\alpha$' is the infection probability\endtabular}
    
	\label{graph:asir1}
\end{figure}

\item \textbf{Turn recovered}. When $a_k$ was \textbf{I}nfected at time $t-1$ it has a chance to turn \textbf{R}ecovered at time $t$. Written as:

\begin{equation}
   \Bigg (  H_{a_k}^{t-1}\to H_{a_k}^{t}  =   I\to R \Bigg \vert  H_{a_k}^{t-1} = I \Bigg )
    \label{eq:recover}
\end{equation}

or simply:

\begin{equation}
     H_{a_k}^{t-1}\to H_{a_k}^{t}  =   I\to R
     \label{eq:recover2}
\end{equation}

\begin{figure}[h]
	\begin{center}
		\includegraphics[scale=0.618]{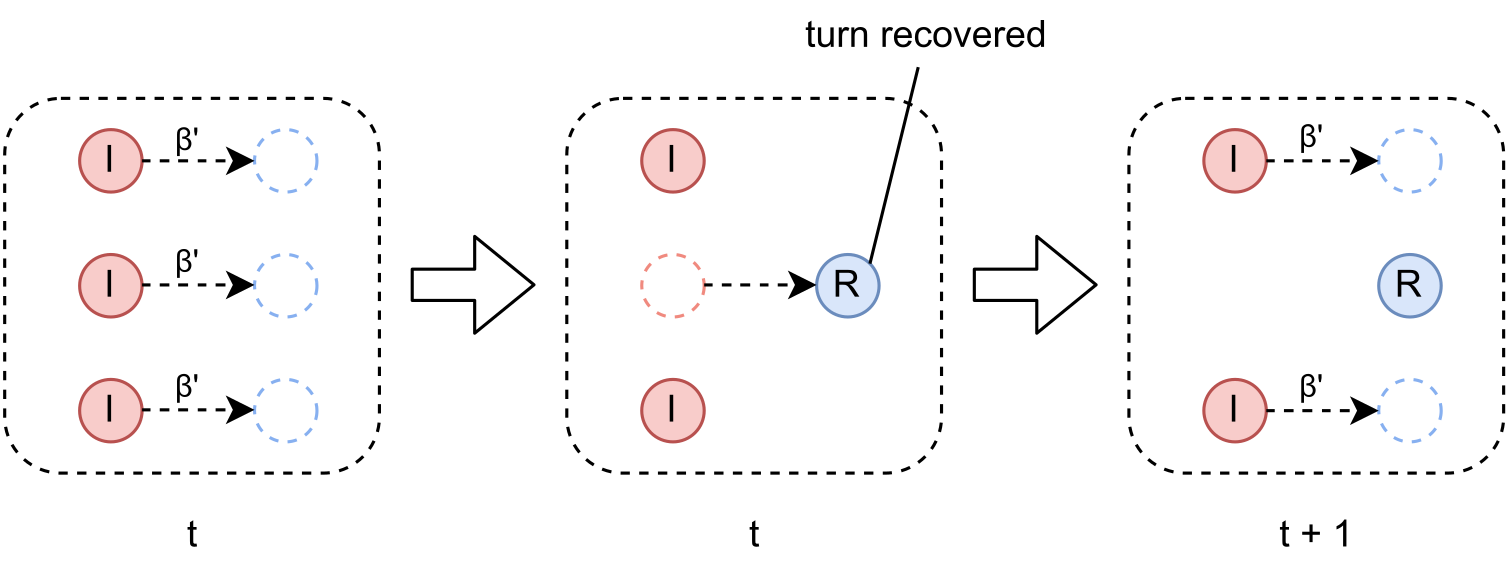}
	\end{center}
  
   \caption[Individual recovery process in the ASIR model; $\beta$' is the recovery probability]
    {\tabular[t]{@{}l@{}}Individual recovery process in the ASIR model. \\ $\beta$' is the recovery probability \endtabular}
    
	\label{graph:asir2}
\end{figure}
  
\end{enumerate}

As mentioned in the introduction, a core idea supporting $ASIR$ to robustly reproduce $SIR$'s prediction is using parameters to control the \textbf{event probability}. In our model, there are only two types of random events that involve transition in individual health state $H_{a_k}^{t}$: 1). \textbf{turn infected}: $\big(H_{a_k}^{t-1}\to H_{a_k}^{t} \big) =  \big( S\to I \big) $ , and 2). \textbf{turn recovered}: $\big(H_{a_k}^{t-1}\to H_{a_k}^{t} \big) =  \big( I\to R \big)$. We are using two parameters $\{\alpha', \beta' \}$ to control their probability in the following way:

\begin{enumerate}
\item $a_k$'s  probability of being infected at $t$ is proportional to the count of infected neighbor, with the ratio equals to parameter $\alpha'$

\begin{equation}
    \frac{ \mathds{P}\big(H_{a_k}^{t-1}\to H_{a_k}^{t} =  S\to I \big) } 
    { \mathlarger{\sum}_{a_j}^{\ a_j \neq a_k} \mathds{1}_{\{  H_{a_j}^t = I, P^{a_j}_t = P^{a_k}_t  \}} } =  \mathds{1}_{\{H_{a_k}^{t-1} = S\}}  \cdot \alpha'
\end{equation}

or simply,
\begin{equation}
    \frac{ \mathds{P}\big(H_{a_k}^{t-1}\to H_{a_k}^{t} =  S\to I  \big) } 
    { \mathlarger{\sum}_{a_j}^{\ a_j \neq a_k} \mathds{1}_{a_j \text{ is infected neighbor of } a_k \text{ at } t}} =  \alpha'
\end{equation}

Notably, here we are making an assumption that: 

\begin{equation}
\mathds{P}\Big( 1 < \alpha' * \mathlarger{\sum}_{a_j}^{ a_j \neq a_k} \mathds{1}_{a_j \text{ is infeacted neighbor of } a_k \text{ at } t} \Big) = 0
\end{equation}

This means we are requiring the agent density at any location to be reasonably low in order to reflect the reality that people almost never stay in a place that has absolute causality in terms of infection.

\item An \textbf{I}nfected $a_k$'s probability of being recovered at any time $t$ is constant, which equals to parameter $\beta'$

\begin{equation}
\mathds{P}(H_{a_k}^{t-1}\to H_{a_k}^{t}  =   I\to R) =  \mathds{1}_{\{H_{a_k}^{t-1} = I\}} \cdot \beta'  
\end{equation}

or simply:

\begin{equation}
    \mathds{P}(H_{a_k}^{t-1}\to H_{a_k}^{t}  =   I\to R) = \beta'    
\end{equation}

\end{enumerate}
Based on the $ASIR$ specification above, we propose the following theorem:

\begin{theorem}
($ASIR$ Robustness theorem) $ASIR$ robustly reproduce the infection curve predicted by $SIR$, i.e. "$\forall t: \mathbb{E}\big [ S_{asir}(t) \big] = S_{sir}(t)$; \ $\mathbb{E}\big[ I_{asir}(t) \big] = I_{sir}(t)$; \ $\mathbb{E}\big[ R_{asir}(t)\big] = R_{sir}(t) $",  if:
\begin{enumerate}
\item $S_{asir}(0) = S_{sir}(0)$; \  $I_{asir}(0) = I_{sir}(0)$; \ $R_{asir}(0) = R_{sir}(0)$, and
\item $\alpha = \alpha ' \cdot \mathds{P}(\text{meetup})  \cdot N $;  $\mathds{P}(\text{meetup})$ is determined by $ASIR$'s transition matrix $T_{\text{map}}$, $N$ is the population size, and
\item $\beta = \beta'$
\end{enumerate}
\label{th:1}
\end{theorem}

We will give the proof in the "Proof of Robustness" section.

\subsection{Proof of Robustness}

Let us give the proof for theorem \ref{th:1}. We write "period $(t, t+k)$" when we refer to the period since $t$ till $t+k$ . We are using "$\Delta S_\text{sir} \rvert_{t}^{t+k} $" to denote the change in $S_{sir}(t)$ during the period $(t, t+k)$:
\begin{align}
\Delta S_\text{sir} \rvert_{t}^{t+k} &= S_\text{sir}(t+k) - S_\text{sir}(t) \\
\Delta I_\text{sir} \rvert_{t}^{t+k} &= I_\text{sir}(t+k) - I_\text{sir}(t) \\
\Delta R_\text{sir} \rvert_{t}^{t+k} &= R_\text{sir}(t+k) - R_\text{sir}(t) 
\end{align}Similarly, we are using "$\Delta S_\text{asir} \rvert_{t}^{t+k} $" to denote the change in $S_{asir}(t)$ during the period $(t, t+k)$:
\begin{align}
\Delta S_\text{asir} \rvert_{t}^{t+k} &= S_\text{asir}(t+k) - S_\text{asir}(t) \\
\Delta I_\text{asir} \rvert_{t}^{t+k} &= I_\text{asir}(t+k) - I_\text{asir}(t) \\
\Delta R_\text{asir} \rvert_{t}^{t+k} &= R_\text{asir}(t+k) - R_\text{asir}(t) \\
\end{align}
We prove theorem \ref{th:1} by proving theorems \ref{th:2}, \ref{th:3}:\\

\begin{figure}[htp]
\begin{center}
\begin{tikzcd}[sep=1em]
\tikzcdset{
arrow style=tikz,
arrows={thick,double},
diagrams={>=stealth}
}
\textbf{Theorem 3}\arrow[r] & \textbf{Theorem 2}\arrow[r] & \textbf{Theorem 1} \\
\end{tikzcd}
\caption[Relationship between theorems: theorem \ref{th:3} implies theorem \ref{th:2}; theorem \ref{th:2} implies theorem \ref{th:1}.] {\tabular[t]{@{}l@{}}Relationship between theorems: theorem \ref{th:3} implies theorem \ref{th:2};\\ theorem \ref{th:2} implies theorem \ref{th:1}. \endtabular}
    
\label{img:dependency}
\end{center}
\end{figure}

\begin{theorem}
$ASIR$ robustly reproduce the infection curve predicted by $SIR$, i.e. "$\forall t: \mathbb{E}\big [ S_{asir}(t) \big] = S_{sir}(t)$; \ $\mathbb{E}\big[ I_{asir}(t) \big] = I_{sir}(t)$; \ $\mathbb{E}\big[ R_{asir}(t)\big] = R_{sir}(t) $",  if for any period $(t, t+k)$:
\begin{enumerate}
\item $\mathbb{E}\big [ \Delta S_\text{asir} \big \rvert_{t}^{t+k} \big] = \Delta S_\text{sir} \rvert_{t}^{t+k} $, and
\item $\mathbb{E}\big [ \Delta I_\text{asir} \big \rvert_{t}^{t+k} \big] = \Delta I_\text{sir} \rvert_{t}^{t+k} $ , and 
\item $\mathbb{E}\big [ \Delta R_\text{asir} \big \rvert_{t}^{t+k} \big] = \Delta R_\text{sir} \rvert_{t}^{t+k} $ 
\end{enumerate}
\label{th:2}
\end{theorem}

\begin{proof}
Take period $(0, t)$, 
\begin{enumerate}
\item $\mathbb{E}\big [ S_\text{asir}(t) \big] = \mathbb{E}\big [ \Delta S_\text{asir} \big \rvert_{0}^{t} \big] + S_\text{asir}(0)= \Delta S_\text{sir} \rvert_{0}^{t} + S_\text{sir}(0) = S_\text{sir}(t) $, and
\item $\mathbb{E}\big [ I_\text{asir}(t) \big] = \mathbb{E}\big [ \Delta I_\text{asir} \big \rvert_{0}^{t} \big] + I_\text{asir}(0)= \Delta I_\text{sir} \rvert_{0}^{t} + I_\text{sir}(0) = I_\text{sir}(t) $, and
\item $\mathbb{E}\big [ R_\text{asir}(t) \big] = \mathbb{E}\big [ \Delta R_\text{asir} \big \rvert_{0}^{t} \big] + R_\text{asir}(0)= \Delta R_\text{sir} \rvert_{0}^{t} + R_\text{sir}(0) = R_\text{sir}(t) $
\end{enumerate}
\end{proof}

\begin{theorem}
For any period $(t, t+k)$: 

\begin{enumerate}
\item $\mathbb{E}\big [ \Delta S_\text{asir} \big \rvert_{t}^{t+k} \big] = \Delta S_\text{sir} \rvert_{t}^{t+k} $
\item $\mathbb{E}\big [ \Delta I_\text{asir} \big \rvert_{t}^{t+k} \big] = \Delta I_\text{sir} \rvert_{t}^{t+k} $
\item $\mathbb{E}\big [ \Delta R_\text{asir} \big \rvert_{t}^{t+k} \big] = \Delta R_\text{sir} \rvert_{t}^{t+k} $
\end{enumerate}

if:

\begin{enumerate}
\item $S_{asir}(0) = S_{sir}(0)$; \  $I_{asir}(0) = I_{sir}(0)$; \ $R_{asir}(0) = R_{sir}(0)$, and
\item $\alpha = \alpha ' \cdot \mathds{P}(\text{meetup})  \cdot N $;  $\mathds{P}(\text{meetup})$ is determined by $ASIR$'s transition matrix $T_{\text{map}}$, $N$ is the population size, and
\item $\beta = \beta'$
\end{enumerate}
\label{th:3}
\end{theorem}

\begin{proof}
In $SIR$, 
\begin{align}
\begin{split}
\Delta S_\text{sir} \rvert_{t}^{t+1} = \int_{t}^{t+1} \frac{d S}{d t} dt &= - \frac{\alpha}{N} \cdot S_\text{sir}(t)  I_\text{sir}(t) \cdot \int_{t}^{t+1} dt \\
&= -\frac{\alpha}{N} \cdot S_\text{sir}(t)  I_\text{sir}(t)
\end{split}\\
\begin{split}\\
\Delta R_\text{sir} \rvert_{t}^{t+1} = \int_{t}^{t+1} \frac{d R}{d t} dt &= \beta \cdot  I_\text{sir}(t) \cdot \int_{t}^{t+1} dt \\
&= \beta \cdot I_\text{sir}(t) 
\end{split}\\
\begin{split}\\
\Delta I_\text{sir} \rvert_{t}^{t+1} =  - \Delta S_\text{sir} \rvert_{t}^{t+1} - \Delta R_\text{sir} \rvert_{t}^{t+1} &= \frac{\alpha}{N} \cdot S_\text{sir}(t)  I_\text{sir}(t) - \beta \cdot I_\text{sir}(t) 
\end{split}
\end{align}

In $ASIR$,

\begingroup
\allowdisplaybreaks
\begin{align}
\begin{split}
\mathbb{E}\big [ \Delta S_\text{asir} \big \rvert_{t}^{t+1} \big]
&=-\lsum_{\scriptscriptstyle  \big\{ a_k \big\vert H_{a_k}^{t} = S \big\} } \lsum_{\scriptscriptstyle   \big\{ a_j \big\vert H_{a_j}^{t} = I \big\} } \mathds{1}_{(P^{a_j}_t = P^{a_k}_t  )} \cdot  \mathds{1}_{(H_{a_k}^{t}\to H_{a_k}^{t+1} =  S\to I  )} \\
&= -\lsum_{\scriptscriptstyle  \big\{ a_k \big\vert H_{a_k}^{t} = S \big\} } \lsum_{\scriptscriptstyle   \big\{ a_j \big\vert H_{a_j}^{t} = I \big\} } \mathds{P}(\text{meetup}) \cdot \alpha'  \\
&= -\alpha' \cdot \mathds{P}(\text{meetup}) \cdot  \mathbb{E}\big[ S_\text{asir}(t)\big] \cdot \mathbb{E}\big[I_\text{asir}(t)\big]  \\
&= -\frac{\alpha}{N}\cdot \mathbb{E}\big[ S_\text{asir}(t)\big] \cdot \mathbb{E}\big[I_\text{asir}(t)\big]
\end{split}\\
\begin{split}\\
\mathbb{E}\big [ \Delta R_\text{asir} \big \rvert_{t}^{t+1} \big]
&=\lsum_{\scriptscriptstyle  \big\{ a_k \big\vert H_{a_k}^{t} = I \big\} }  \mathds{1}_{H_{a_k}^{t-1}\to H_{a_k}^{t}  =   I\to R} \\
&= \lsum_{\scriptscriptstyle  \big\{ a_k \big\vert H_{a_k}^{t} = I \big\} }  \beta' \\
&= \mathbb{E}\big[ I_\text{asir}(t)\big] \cdot \beta' 
\end{split}\\
\begin{split}\\
\mathbb{E}\big [ \Delta I_\text{asir} \big \rvert_{t}^{t+1} \big]
&= -\mathbb{E}\Big [ \Delta S_\text{asir} \big \rvert_{t}^{t+1} \Big] - \mathbb{E}\Big [ \Delta R_\text{asir} \big \rvert_{t}^{t+1} \Big]  \\
&= \frac{\alpha}{N}\cdot \mathbb{E}\big[ S_\text{asir}(t)\big] \cdot \mathbb{E}\big[I_\text{asir}(t)\big] - \mathbb{E}\big[ I_\text{asir}(t)\big] \cdot \beta'
\end{split}
\end{align}
\endgroup

When $t = 0$,
\begingroup
\allowdisplaybreaks
\begin{align}
\begin{split}
\mathbb{E}\Big [ S_\text{asir}(1)  \Big]  &= \mathbb{E}\big[ S_\text{asir}(0) \big]+ \mathbb{E}\Big [ \Delta S_\text{asir} \Big \rvert_{0}^{1} \Big] \\
&=  S_\text{asir}(0) + \frac{\alpha}{N}\cdot \mathbb{E}\big[ S_\text{asir}(0)\big] \cdot \mathbb{E}\big[I_\text{asir}(0)\big] \\
&= S_\text{sir}(0) +  \frac{\alpha}{N}\cdot  S_\text{sir}(0) \cdot I_\text{sir}(0) \\
&= S_\text{sir}(0) +  \Delta S_\text{sir} \rvert_{0}^{1} \\
&= S_\text{sir}(1)
\end{split}\\
\begin{split}\\
\mathbb{E}\Big [ R_\text{asir}(1)  \Big]  &= \mathbb{E}\big[ R_\text{asir}(0) \big]+ \mathbb{E}\Big [ \Delta R_\text{asir} \Big \rvert_{0}^{1} \Big] \\
&= \mathbb{E}\big[ R_\text{asir}(0) \big] +  \mathbb{E}\big[ I_\text{asir}(0)\big] \cdot \beta' \\
&= R_\text{asir}(0) + \beta' \cdot I_\text{asir}(0)  \\
&= R_\text{sir}(0) + \beta \cdot  I_\text{sir}(0) \\
&= R_\text{sir}(0) + \Delta R_\text{sir} \rvert_{0}^{1} \\
&= R_\text{sir}(1)
\end{split}\\
\begin{split}\\
\mathbb{E}\Big [ I_\text{asir}(1)  \Big]  &= N - \mathbb{E}\Big [ S_\text{asir}(1)  \Big] - \mathbb{E}\Big [ R_\text{asir}(1) \Big] \\
&= N -  S_\text{sir}(1) -  R_\text{sir}(1) \\
&= I_\text{sir}(1)
\end{split}
\end{align}
\endgroup

When $t > 0$, assume $\mathbb{E}\big [ S_\text{asir}(t-1)  \big] =  S_\text{sir}(t-1) $, \ $\mathbb{E}\big [ I_\text{asir}(t-1)  \big] =  I_\text{sir}(t-1) $, and \ $\mathbb{E}\big [ R_\text{asir}(t-1)  \big] =  R_\text{sir}(t-1) $ ,

\begingroup
\allowdisplaybreaks
\begin{align}
\begin{split}
\mathbb{E}\Big [ S_\text{asir}(t)  \Big]  &= \mathbb{E}\Big [S_\text{asir}(t-1)\Big ] + \mathbb{E}\Big [ \Delta S_\text{asir} \Big \rvert_{t-1}^{t} \Big] \\
&=   \mathbb{E}\Big [S_\text{asir}(t-1)\Big ]  - \frac{\alpha'}{N}\cdot \mathbb{E}\big[ S_\text{asir}(t-1)\big] \cdot \mathbb{E}\big[I_\text{asir}(t-1)\big]\\
&=  S_\text{sir}(t-1) - \frac{\alpha}{N} \cdot S_\text{sir}(t-1) \cdot I_\text{sir}(t-1)\\
&= S_\text{sir}(t)  
\end{split}\\
\begin{split}\\
\mathbb{E}\Big [ R_\text{asir}(t)  \Big]  &=   \mathbb{E}\big[ R_\text{asir}(t-1) \big]+ \mathbb{E}\Big [ \Delta R_\text{asir} \Big \rvert_{t-1}^{t} \Big]  \\
&= \mathbb{E}\big[ R_\text{asir}(t-1) \big]+ \mathbb{E}\big[ I_\text{asir}(t-1)\big] \cdot \beta'  \\
&= S_\text{sir}(t-1) + I_\text{sir}(t-1) \cdot \beta \\
&= R_\text{sir}(t)  
\end{split}\\
\begin{split}\\
\mathbb{E}\Big [ I_\text{asir}(t)  \Big]  &= N - \mathbb{E}\Big [ S_\text{asir}(t)  \Big] - \mathbb{E}\Big [ R_\text{asir}(t) \Big] \\
&= N -  S_\text{sir}(t) -  R_\text{sir}(t) \\
&= I_\text{sir}(t)
\end{split}
\end{align}
\endgroup

By mathematical induction, for any period $(t, t+k)$: $\mathbb{E}\big [ \Delta S_\text{asir} \big \rvert_{t}^{t+k} \big] = \Delta S_\text{sir} \rvert_{t}^{t+k} $, \  
$\mathbb{E}\big [ \Delta I_\text{asir} \big \rvert_{t}^{t+k} \big] = \Delta I_\text{sir} \rvert_{t}^{t+k} $, \ $\mathbb{E}\big [ \Delta R_\text{asir} \big \rvert_{t}^{t+k} \big] = \Delta R_\text{sir} \rvert_{t}^{t+k} $.

\end{proof}

Proof of Robustness:
\begin{proof}
Theorem \ref{th:3} implies theorem \ref{th:2}; theorem \ref{th:2} implies theorem \ref{th:1}.
\end{proof}

\section{Results}
In this section, we will discuss the outcome of our \textit{imperfect}\footnote{This happened because those implementations were crafted before we mature the model design and completed the \textit{Proof of Robustness}.} $ASIR$ implementations in two multiagent programming environments: \textbf{GAMA} and \textbf{Agents.jl}\footnote{Agents.jl is more flexible and programmable than GAMA. We shift to Agents.jl to examine if the robustness we observed in GAMA retains under different environment.}.

\subsection{$ASIR$ implementation in \textbf{GAMA} }
Our $ASIR$ implementation in \textbf{GAMA} uses a small map ($50 \times 50$ grid) and large population size ($N=500$.) We observe a perfect reproduction of the coresponding $SIR$'s infection curve. We think this is because $\mathds{P}(\text{meetup})$ and $\frac{\text{step size}}{\text{map size}}$  is high enough to quickly reach $T_{\text{map} }$'s stationary distribution.\\

\begin{figure}[h]
	\begin{center}
		\includegraphics[scale=0.18]{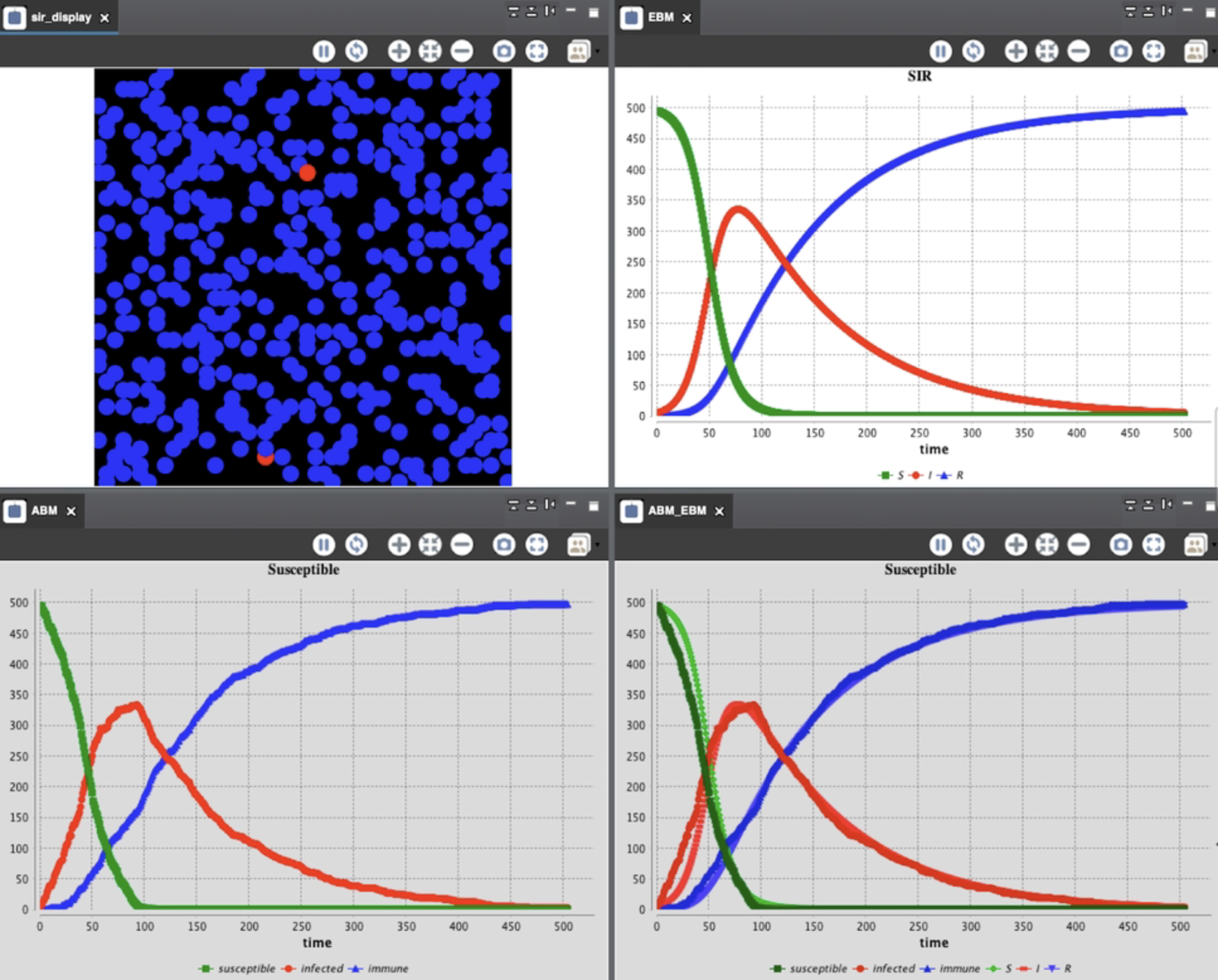}
	\end{center}
\caption[$ASIR$ implementation on \textbf{GAMA} perfectly reproduces the infection curve of the given $SIR$]
    {\tabular[t]{@{}l@{}}$ASIR$ implementation on \textbf{GAMA} perfectly reproduce the infection curve \\ of the given $SIR$. On bottom-left is $ASIR$'s infection curve; on top-right \\ is $SIR$'s infection curve; on bottom-right is a comparison between the two.\endtabular}
	\label{graph:gama}
\end{figure}

\subsection{$ASIR$ implementation in \textbf{Agents.jl} }
However, our $ASIR$ implementation in \textbf{Agents.jl} with a larger map ($100 \times 100$ grid) and the same population size ($N = 100$) fails to robustly reproduce the infection curve of the same benchmark $SIR$ we used in \textbf{GAMA}. (As Figure \ref{graph:jl} illustrates.) The initial \textbf{I}nfected agents recover before infecting enough \textbf{S}usceptible agents; so both $I_{asir}(t)$ and $R_{asir}(t)$ full curves remain flat to 0.

\begin{figure}[h]
	\begin{center}
		\includegraphics[scale=0.25]{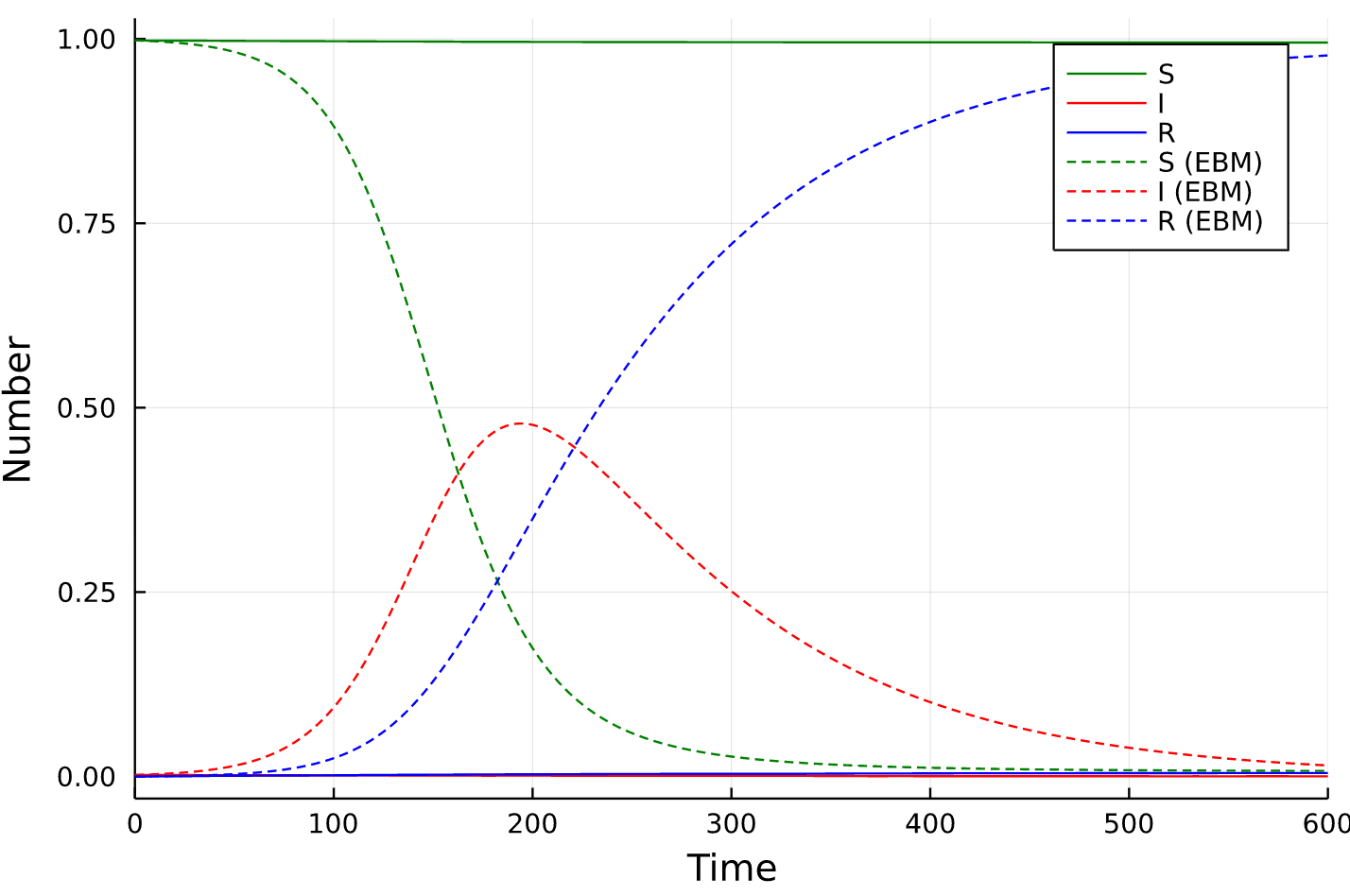}
	\end{center}
\caption[$ASIR$ implementation on \textbf{Agents.jl} fails to reproduce the infection curve of the given $SIR$; both $I_{asir}(t)$ and $R_{asir}(t)$ full curves remain flat to 0.]
    {\tabular[t]{@{}l@{}}$ASIR$ implementation on \textbf{Agents.jl} fails to reproduce the infection curve\\ of the given $SIR$; both $I_{asir}(t)$ and $R_{asir}(t)$ full curves remain flat to 0.\endtabular}
	\label{graph:jl}
\end{figure}
When we play back the trajectory of agents, we find agents are loosely distributed, and their trajectories have few intersections ($\mathds{P}(\text{meetup}) \approx 0 $, $\frac{\text{step size}}{\text{map size} } \approx 0 $). 

We think this is a counterexample where $T_\text{map}$'s stationary distribution can not be reached in finite steps. A primitive explanation is: when $\frac{\text{step size}}{\text{map size}}\approx 0$, positions that are too far away are \textbf{nearly inaccessible to each other in finite steps}; for example, an agent will take nearly infinite steps to move from the top-left corner to the bottom-right corner. This breaks one of the necessary conditions that "\textit{there is only one communication class}" for $T_\text{map}$'s stationary distribution to exist in finite steps.

\section{Discussion}

In this section, we will discuss the primary imperfection of our $ASIR$ implementations. 

The primary imperfection of our implementations is: \textit{they does not guarantee agents to reach the $T_\text{map}$'s stationary distribution at time $t=0$.} This is because we did not implement the \textbf{Move} behavior in a direct form of transition matrix $T_\text{map}$, but as "\textit{at each timestamp $t$, randomly choose one direction from} $(\text{up}, \text{down}, \text{left}, \text{right})$, \textit{then move one unit distance accordingly}." Though theoretically, this implementation of the \textbf{Move} behavior has a transition matrix $T_{\text{map} }$ representation which has a stationary distribution. The problem is: we cannot determine, after which $t$, agents' positions will reach $T_{\text{map} }$'s stationary distribution. Also, we have to estimate but not deduce the value of $\mathds{P}(\text{meetup})$.

As a correction of our present approach, we would implement the \textbf{Move} behavior in direct form of $T_\text{map}$ so that we can 1). guarantee each agents' position will state reach $T_\text{map}$'s stationary distribution at $t=0$, 2). deduce the exact value of $\mathds{P}(\text{meetup})$.

\section{Conclusion}

We propose an agent-based $SIR$ model, $ASIR$ , that achieves the following interesting properties:
\begin{enumerate}
\item  Parameters of $ASIR$ can be deduced from (i.e. written as a determinate expression of) parameters of $SIR$.
\item $ASIR$ robustly reproduce the infection curve predicted by $SIR$ (i.e. \textbf{the expectation values of the population size} of being susceptible, infected and recovered equal to $SIR$'s prediction.)
\end{enumerate}

We validated $ASIR$'s properties by giving: 1). a proof of robustness, 2). two implementations in \textbf{GAMA} and \textbf{Agents.jl}.

There are two interesting directions to extend our work:
\begin{enumerate}
\item Find robust agent-based representations for other compartmental models ($SEIR$, $SEIS$, ...) We could start with describing new types of states and behaviors in the algebraic language we use (like Table \ref{tab:state-reference}) and give a proof of robustness.
\item Construct the transition matrix $T_\text{map}$ based on real data, check if the deduced value of $\alpha'$ comply with intuition. We could start with finding a calibrated $SIR$ model in the literature, where both real \textit{infection data} and \textit{transportation data} are applicable; then use MCMC algorithms like Gibbs sampling to construct the according transition matrix $T_\text{map}$.
\end{enumerate}


\newpage
\singlespacing
\bibliographystyle{IEEEtran}
\bibliography{references}


\end{document}